\documentclass[letterpaper,abstract=true,DIV=14,11pt]{scrartcl}
\usepackage[margin=1in]{geometry} %
\usepackage{algorithm}
\usepackage[noEnd,commentColor=black]{algpseudocodex}
\usepackage{amsmath}
\usepackage{amssymb}
\usepackage{amsthm}
\usepackage[mathscr]{euscript}
\usepackage[shortlabels,inline]{enumitem}
\usepackage{graphicx} 
\usepackage{subcaption}
\usepackage{thmtools}
\usepackage{todonotes}
\usepackage{xcolor}
\usepackage{microtype}
\usepackage{forest} 

\usepackage{lineno}

\usepackage[normalem]{ulem}

\usepackage{tcolorbox}

\usepackage{upgreek}
\usepackage[colorlinks]{hyperref}
\usepackage[capitalize]{cleveref}

\crefname{equation}{eq.}{eqs.} 
\crefname{enumi}{}{} 
\crefname{icase}{case}{cases}
\crefname{ipart}{part}{parts}
\crefname{iprop}{property}{properties}
\crefname{iinv}{invariant}{invariants}

\DeclareMathOperator{\Av}{Av}
\DeclareMathOperator{\ex}{ex}

\usepackage[blocks]{authblk}
\usepackage[normalem]{ulem}

\usepackage{tikz}
\usetikzlibrary{calc}

\tikzset{
	point/.style={circle, fill, inner sep=1.5pt},
	smallpoint/.style={point, inner sep=1.2pt},
	tinypoint/.style={point, inner sep=1pt},
	hlbox/.style={fill, {white!90!black}},
	subrect/.style={draw, fill={white!80!cyan}}, 
	msrect/.style=subrect 
}

\newtheorem{theorem}{Theorem}[section]
\newtheorem{lemma}[theorem]{Lemma}

\newtheorem{claim}[theorem]{Claim}

\newtheorem{corollary}[theorem]{Corollary}
\newtheorem{observation}[theorem]{Observation}
\theoremstyle{definition}

\title{An Optimal Algorithm for Sorting Pattern-Avoiding Sequences}

\author{Michal Opler\thanks{This work was co-funded by the European Union under the project Robotics and advanced industrial production (reg. no. CZ.02.01.01/00/22\_008/0004590).}}
\affil{Czech Technical University in Prague, Czech Republic}

\date{}

\begin{document}

\maketitle

\begin{abstract}
We present a deterministic comparison-based algorithm that sorts sequences avoiding a fixed permutation $\pi$ in linear time, even if $\pi$ is a priori unkown.
Moreover, the dependence of the multiplicative constant on the pattern $\pi$ matches the information-theoretic lower bound.
A crucial ingredient is an algorithm for performing efficient multi-way merge based on the Marcus-Tardos theorem.
As a direct corollary, we obtain a linear-time algorithm for sorting permutations of bounded twin-width.
\end{abstract}

\section{Introduction}\label{sec:introduction}
The information-theoretic lower bound on sorting states that any comparison-based sorting algorithm must take $\Omega(n \log n)$ time in the worst case for sorting inputs of length $n$, see~\cite{Cormen}.
Naturally, it invites the question about complexity of sorting restricted sets of inputs, and this problem has indeed been studied for decades.
In this paper, we resolve the complexity of sorting input sequences restricted by avoiding a fixed permutation pattern.

We distinguish between arbitrary sequences and permutations that contain no repeated values.
A sequence $S = s_1, \dots, s_n$ \emph{contains} a permutation $\pi = \pi_1, \dots, \pi_k$ if there are indices $i_1 < \dots < i_k$ such that $s_{i_j} < s_{i_\ell}$ if and only if $\pi_j < \pi_\ell$ for all $j,\ell \in [k]$.
We say that the subsequence $s_{i_1}, \dots, s_{i_k}$ of $S$ is \emph{order-isomorphic} to~$\pi$.
Otherwise, $S$ avoids~$\pi$.

\paragraph{Lower bounds.}
A standard argument shows that any comparison-based algorithm for sorting a set of permutations $\mit \Gamma$, each of length $n$, must make at least $\Omega(\log |{\mit \Gamma}| + n)$ comparisons in the worst case.
Fredman~\cite{Fredman} showed that $O(\log|{\mit \Gamma}| + n)$ comparisons suffice if we have complete knowledge of $\mit\Gamma$ and unrestricted computational power.
Specifically, Fredman's argument proves the existence of a decision tree of depth $O(\log|{\mit \Gamma}| + n)$ that correctly sorts~$\mit\Gamma$.
In order to apply these bounds on pattern-avoiding sequences, it invites the natural question of how many pattern-avoiding permutations of given length are there.
In fact, it has been a long-standing open problem in the world of pattern-avoiding permutations whether the number of $\pi$-avoiding permutations of length $n$ grows at most exponentially in~$n$.
This question, known as the \emph{Stanley-Wilf conjecture}, was eventually answered positively by Marcus and Tardos~\cite{MarcusTardos} building on the works of Klazar~\cite{Klazar} and Füredi and Hajnal~\cite{FurediHajnal}.
In fact, the limit $s_\pi = \lim_{n \to \infty} \sqrt[n]{|\Av_n(\pi)|}$ exists~\cite{Arratia99}, where $\Av_n(\pi)$ denotes the set of all $\pi$-avoiding permutations of length~$n$; $s_\pi$ is known as the \emph{Stanley-Wilf limit} of~$\pi$.
It follows that $\pi$-avoiding permutations of length $n$ can be sorted by a decision tree of height $O((\log s_\pi + 1) \cdot n)$ and simultaneously, any comparison-based algorithm sorting them must make at least $\Omega((\log s_\pi + 1) \cdot n)$ comparisons in the worst case.

The Stanley-Wilf limit is closely related to a different constant from the extremal theory of matrices.
Let $\ex_\pi(n)$ be the maximum number of 1-entries in a $\pi$-avoiding $n \times n$ binary matrix.
The Marcus-Tardos theorem~\cite{MarcusTardos}, previously known as the \emph{Füredi-Hajnal conjecture}~\cite{FurediHajnal}, states that $\ex_\pi(n) = O_\pi(n)$.
Cibulka~\cite{Cibulka2009} observed that the limit $c_\pi = \lim_{n \to \infty} \frac{1}{n} \ex_\pi(n)$ exists; $c_\pi$ is known as the \emph{Füredi-Hajnal limit} of $\pi$.
Moreover, Cibulka showed that the limits $c_\pi$ and $s_\pi$ are polynomially related and thus, we can use $\log s_\pi$ and $\log c_\pi$ interchangeably within asymptotic notation.
In particular, the number of comparisons made by any comparison-based algorithm sorting $\pi$-avoiding sequences is at least $\Omega((\log c_\pi + 1) \cdot n)$.

Unfortunately, our knowledge about the Füredi-Hajnal (and Stanley-Wilf) limits for different patterns is rather sparse.
Apart from few specific examples, Fox~\cite{jfox} showed that $c_\pi \in 2^{O(k)}$ for all patterns $\pi$ and $c_\pi \in 2^{\Omega(k^{1/4})}$ for some patterns $\pi$ where $k$ denotes the length of $\pi$.

\paragraph{Pattern-specific algorithms.}
One approach for sorting $\pi$-avoiding inputs is to exploit their structure given by~$\pi$ explicitly to design an efficient sorting algorithm.
Knuth~\cite{Knuth3} was perhaps the first to do so when he observed that a permutation can be sorted by a single pass through a stack if and only if it avoids the pattern $2,3,1$.
This characterization naturally provides a linear algorithm for sorting $(2,3,1)$-avoiding permutations.
Up to symmetry, there is only one other pattern of length three, the permutation $1,2,3$; and Arthur~\cite{Arthur} gave $O(n)$-time algorithm for sorting $(1,2,3)$-avoiding permutations.
In fact, Arthur~\cite{Arthur} devised a more general framework which implies algorithms for sorting $\pi$-avoiding sequences
\begin{enumerate*}[label=(\roman*)]
\item in $O(n)$ time when $\pi \in \{(1,2,3,4), (1,2,4,3), (2,1,4,3)\}$, and 
\item in $O(n \log \log \log n)$ time when $\pi \in \{(1,3,2,4), (1,3,4,2), (1,4,2,3), (1,4,3,2)\}$.
\end{enumerate*}
Up to symmetry, the pattern $2,4,1,3$ remained the only pattern of length four for which Arthur~\cite{Arthur} was unable to improve the standard $O(n \log n)$-time sorting algorithms.

\paragraph{Pattern-agnostic algorithms.}
However, a completely different approach proved lately fruitful.
Instead of designing a specific algorithm depending on the avoided pattern, a general-purpose sorting algorithm is used, and the pattern-avoidance appears merely in the analysis of its runtime.
This approach offers two main advantages over pattern-specific algorithms -- it covers all choices of the avoided pattern $\pi$ simultaneously and moreover, the runtime on a specific input is bounded by the patterns it avoids even though the algorithm itself has no knowledge of them.

Chalermsook, Goswami, Kozma, Mehlhorn, and Saranurak~\cite{FOCS15} were the first to consider this approach with their focus being on binary search trees (BSTs).
They analyzed the performance of an online BST balancing strategy \textsc{GreedyFuture}~\cite{DHIKP} and showed that with its use a $\pi$-avoiding sequence can be sorted in $2^{\alpha(n)^{O(k)}} \cdot n$ time where $k$ is the length of $\pi$ and $\alpha(\cdot)$ is the inverse-Ackermann function.
Later, Kozma and Saranurak~\cite{smooth1} introduced a new heap data structure \textsc{SmoothHeap} and showed that a heap sort with \textsc{SmoothHeap} reaches the same asymptotic bound on runtime while being much easier to implement than \textsc{GreedyFuture}.
Recently, the bound on sorting with \textsc{GreedyFuture} was improved to $2^{O(k^2) + (1 + o(1))\alpha(n)} \cdot n$ by Chalermsook, Pettie and Yingchareonthawornchai~\cite{ChalermsookPettieY}.
Tighter bounds on sorting with \textsc{GreedyFuture} are known only for specific sets of inputs such as $O(nk^2)$ for $(k, \dots, 1)$-avoiding permutations~\cite{BSTdecomp} or $O(n \log k)$ for the so-called $k$-decomposable permutations~\cite{GoyalGupta}.

\paragraph{BSTs and dynamic optimality.}

Arguably, the most famous open question related to sorting is the \emph{dynamic optimality conjecture}~\cite{ST85} asking whether the splay tree is an instance-optimal online BST balancing strategy.
More recently, the conjecture is usually rephrased as asking whether there exists \textit{any} instance-optimal online BST strategy with \textsc{GreedyFuture} being the other prime candidate alongside splay trees.
Observe that any instance-optimal balancing strategy must also work within the bounds proved for \textsc{GreedyFuture}~\cite{FOCS15,BSTdecomp,ChalermsookPettieY} and thus, any such strategy can be used to sort pattern-avoiding sequences in near-linear time.
However, none of these results implies that an instance-optimal BST strategy must sort pattern-avoiding sequences in linear time, and it has been raised as an open question~\cite{FOCS15} whether this is, in fact, a corollary of dynamic optimality.
Berendsohn, Kozma and Opler~\cite{BerensohnKO2023} resolved this question positively by showing that any instance-optimal BST balancing strategy sorts $\pi$-avoiding sequences in time $O(c_\pi^2 \, n)$.
They construct an offline sequence of BST modifications that accommodates a given $\pi$-avoiding input with linear cost.
However, an $\Omega(n \log n)$ time is still needed for the computation of this sequence and therefore, this result in itself does not imply the existence of a linear-time algorithm for sorting pattern-avoiding sequences.

\paragraph{Twin-width.}
Lastly, there is a close connection between pattern avoidance and the recently emerging notion of twin-width.
Without going into too much detail, twin-width is a structural parameter introduced for permutations by Guillemot and Marx~\cite{GM_PPM}, and later extended to graphs and arbitrary relational structures by Bonnet, Kim, Thomassé and Watrigant~\cite{BonnetKTW22}.
Guillemot and Marx~\cite{GM_PPM} showed that the twin-width of any $\pi$-avoiding permutation is at most $O(c_\pi)$ and conversely, any permutation of twin-width $d$ avoids a certain pattern $\pi$ with $c_\pi = 2^{O(d)}$.
As a result, bounded twin-width can been exploited to prove properties of pattern-avoiding permutations and sequences~\cite{BonnetBourneufEtAl2024,BerensohnKO2023}.

\subsection{Our result}

Our main result is an asymptotically optimal sorting algorithm for pattern-avoiding sequences that matches the information-theoretic lower bound for each $\pi$ up to a global constant.

\begin{theorem}\label{thm:main-result}
There is an algorithm that sorts a $\pi$-avoiding sequence $S$ of length $n$ in $O((\log c_\pi + 1) \cdot  n)$ time even if $\pi$ is a priori unknown.
\end{theorem}

The algorithm is built from two components.
The first one is an efficient multi-way merge procedure capable of merging up to $\frac{n}{\log n}$ sequences with $n$ elements in total, as long as the sequences originate from a $\pi$-avoiding input.
Its analysis is heavily based on the Marcus-Tardos theorem~\cite{MarcusTardos}.

The second component is an efficient algorithm to sort a large set of short sequences.
Naturally, many sequences in such a set must be pairwise order-isomorphic, which suggests that we could try grouping them into equivalence classes and then actually sorting only one sequence from each class.
Unfortunately, computing the equivalence classes is in some sense as hard as sorting itself.
Instead, we take a different approach of guessing a shallow decision tree for sorting $\pi$-avoiding permutations guaranteed by the result of Fredman~\cite{Fredman}.

\cref{thm:main-result} follows by first cutting up the input into sequences of length roughly $\log \log \log n$ which we can efficiently sort using the optimal decision tree.
The sorted sequences then serve as a basis for a bottom-up merge sort using the efficient multi-way merge.

Furthermore, let us point out that \cref{thm:main-result} achieves the information-theoretic lower bound for arbitrary $\pi$-avoiding sequences, not only permutations.
This is in sharp contrast with previous works~\cite{FOCS15,BerensohnKO2023} where the generalization beyond permutations comes with a price of worsening the dependence on $\pi$.
In order to achieve this, we strengthen a combinatorial result of Cibulka~\cite{Cibulka2009} and show that the number of $\pi$-avoiding $n \times n$  binary matrices with $n$ 1-entries is at most $c_\pi^{2n + O(1)}$.

\medskip

Recall that there is a connection between pattern-avoiding permutations and permutations of bounded twin-width.
In particular, any permutation of twin-width $d$ avoids a certain pattern $\pi$, the so-called $d \times d$-grid with Füredi-Hajnal limit $c_\pi \in 2^{O(d)}$.
It follows that the algorithm of \cref{thm:main-result} can also be used to sort permutations of bounded twin-width in linear time, without having access to a decomposition of small width.

\begin{corollary}\label{cor:sort-tww}
There is a comparison-based algorithm that sorts a permutation of length $n$ and twin-with $d$ in $O(d \cdot  n)$ time even if $d$ is a priori unknown.
\end{corollary}

Guillemot and Marx~\cite{GM_PPM} devised a fixed-parameter tractable (fpt-) algorithm that computes a twin-width decomposition of width $2^{O(d)}$ for any permutation $\tau$ of twin-width $d$ and length $n$ in $O(n)$ time.
Crucially, their algorithm assumes access to $\tau$ in sorted order by values, and otherwise its running time degrades to $O(n \log n)$ due to the extra sorting step.
However, \cref{cor:sort-tww} speeds up the sorting step to linear time, and thus a $2^{O(d)}$-wide decomposition can be computed in $O(d \cdot n)$ time even in the comparison-based model.
As a consequence, the running time of several fpt-algorithms can be improved from $O_d(n \log n)$ to $O_d(n)$ when the input permutation of twin-width $d$ is given as a comparison oracle without any precomputed decomposition.
This improvement concerns
\begin{enumerate*}[label=(\roman*)]
\item the permutation pattern matching algorithm by Guillemot and Marx~\cite{GM_PPM},
\item the algorithm for first-order model checking by Bonnet, Kim, Thomassé and Watrigant~\cite{BonnetKTW22},
\item the algorithm for factoring pattern-avoiding permutations into separable ones by Bonnet, Bourneuf, Geniet and Thomassé~\cite{BonnetBourneufEtAl2024}, and
\item algorithms for various optimization problems on pattern-avoiding inputs by Berendsohn, Kozma and Opler~\cite{BerensohnKO2023}.
\end{enumerate*}

\paragraph{Structure of the paper.}
In \cref{sec:prelims}, we introduce the necessary notions and review relevant facts about the Füredi-Hajnal limits.
An efficient algorithm for multi-way merging is presented in \cref{sec:merge}, followed by an efficient algorithm for sorting a large number of short sequences in \cref{sec:sorting-short}.
Finally, \cref{thm:main-result} is proved in \cref{sec:algo} by combining the two algorithms.

\section{Preliminaries}\label{sec:prelims}

\paragraph{Permutations and avoidance.}

A \emph{permutation} $\pi$ of length~$n$ is a sequence $\pi_1, \dots, \pi_n$ in which each number from the set $[n] = \{1, \dots ,n\}$ appears exactly once.
Alternatively, we can represent~$\pi$ as an $n \times n$ binary permutation matrix $M_\pi$ where the cell in the $i$th column and $j$th row, denoted $M_\pi(i,j)$, contains a 1-entry if and only if $j = \pi_i$.
Note that for consistency with permutations, we number rows of matrices bottom to top.
A \emph{permutation sequence} is any sequence $s_1, \dots, s_n$ of pairwise distinct values.
We say that two sequences (not necessarily permutations) $s_1, \dots, s_n$ and $t_1, \dots, t_n$ are \emph{order-isomorphic} if for each $i,j \in [n]$, we have $s_i < s_j$ if and only if $t_i < t_j$.
We say that a sequence $S$ \emph{contains} a permutation $\pi$ if some subsequence of $S$ is order-isomorphic to $\pi$.
Otherwise, we say that $S$ \emph{avoids}~$\pi$. 
See \Cref{fig:perm-reprs}.

\begin{figure}
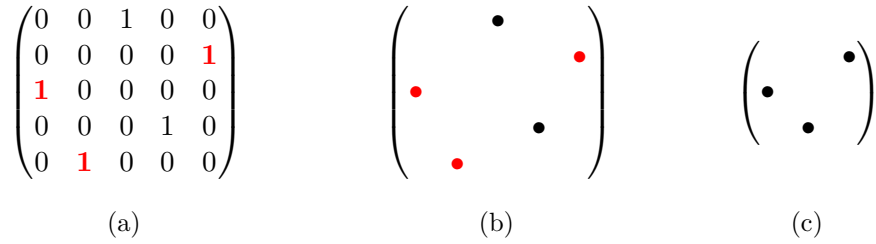

	\centering
	\begin{subfigure}[c]{0.3\textwidth}%
		\centering%
		\[	\begin{pmatrix}
			0&0&1&0&0\\
			0&0&0&0&\textcolor{red}{\textbf{1}}\\				
			\textcolor{red}{\textbf{1}}&0&0&0&0\\
			0&0&0&1&0\\
			0&\textcolor{red}{\textbf{1}}&0&0&0
		\end{pmatrix}
		\]
	\end{subfigure}%
	\begin{subfigure}[c]{0.3\textwidth}
	\centering%
\[	\begin{pmatrix}
	 & &\bullet& & \\
	 & & & &\textcolor{red}{\bullet}\\				
	\textcolor{red}{\bullet}&&&&\\
	 &&&\bullet&\\
	&\textcolor{red}{\bullet}&&&
\end{pmatrix}
\]
	\end{subfigure}%
	\begin{subfigure}[c]{0.2\textwidth}
	\centering%
\[	\begin{pmatrix}
	& & \bullet\\				
	\bullet&&\\
	&\bullet&
\end{pmatrix}
\]
	\end{subfigure}%
	\vspace{1mm}
	
	\begin{subfigure}{0.3\textwidth}
		\caption{}\label{fig:perm-reprs-a}
	\end{subfigure}%
	\begin{subfigure}{0.3\textwidth}
		\caption{}\label{fig:perm-reprs-b}
	\end{subfigure}%
	\begin{subfigure}{0.2\textwidth}
		\caption{}\label{fig:perm-reprs-c}
	\end{subfigure}
	
	\caption{The permutation $3,1,5,2,4$ as a matrix (a) and using bullets for 1-entries and blanks for 0-entries following convention~(b); an occurrence of the pattern $2,1,3$~(c) is highlighted.}\label{fig:perm-reprs}
\end{figure}

Additionally, we define avoidance for binary matrices.
We say that a binary matrix $M$ \emph{contains} a binary matrix $P$ if $P$ can be obtained from $M$ by deleting rows or columns and changing an arbitrary subset of 1-entries to zero.
Otherwise, $M$ \emph{avoids}~$P$.

\paragraph{Furedi Hajnal limits.}
For a permutation $\pi$, let $\ex_\pi(n)$ denote the maximum number of 1-entries in a $\pi$-avoiding $n\times n$ binary matrix.
The celebrated \emph{Marcus-Tardos theorem}~\cite{MarcusTardos}, previously known as the \emph{Füredi-Hajnal conjecture}, states that $\ex_\pi(n)$ is at most $O(n)$ for each fixed $\pi$.
It was observed by Cibulka~\cite{Cibulka2009} that the limit $c_\pi = \lim_{n \to \infty} \frac{1}{n} \ex_\pi(n)$ exists.
This limit is known as the \emph{Füredi-Hajnal limit} of $\pi$.
Furthermore, we have $\ex_\pi(n) \le c_\pi \cdot n$ for all $n$ due to the superadditivity of $\ex_\pi(\cdot)$~\cite{Cibulka2009}, which easily extends to rectangular matrices.

\begin{lemma}\label{lem:marcus-tardos-rect}
Every $m \times n$ binary matrix with strictly more than $c_\pi \cdot \max(m,n)$ 1-entries contains~$\pi$.
\end{lemma}

The \emph{Stanley-Wilf conjecture} states that the number of $\pi$-avoiding permutations of length $n$ grows at most single-exponentially in $n$ .
Let~$\Av_n(\pi)$ be the set of all $\pi$-avoiding permutations of length $n$. The \emph{Stanley-Wilf limit} $s_\pi$ of $\pi$ is defined as
$s_\pi = \lim_{n \rightarrow \infty} \sqrt[n]{|\Av_n(\pi)|}.$

Klazar~\cite{Klazar} showed that the Stanley-Wilf conjecture is implied by the Füredi-Hajnal conjecture.
Importantly, Cibulka~\cite{Cibulka2009} proved that $s_\pi$ and $c_\pi$ are polynomially related, in particular $c_\pi \in O(s_\pi^{4.5})$ and $s_\pi \in O(c_\pi^2)$.
This is important for our purposes, since we can use $\log c_\pi$ and $\log s_\pi$ interchangeably within asymptotic notation.
In particular, from now on we use mostly only the Füredi-Hajnal limit $c_\pi$.

By applying the information-theoretic lower bound together with the relationship between $s_\pi$ and $c_\pi$, we obtain the following observation showing that the dependence on $\pi$ in \cref{thm:main-result} is asymptotically optimal.

\begin{observation}
Any comparison-based sorting algorithm must make at least $\Omega((\log c_\pi +1) \cdot n)$ comparisons in the worst case on $\pi$-avoiding sequences. 	
\end{observation}

\section{Efficient multi-way merge}
\label{sec:merge}

In this section, we show how we can leverage the Marcus-Tardos theorem~\cite{MarcusTardos} to implement an efficient multi-way merging procedure that can handle up to $\frac{n}{\log n}$ many pattern-avoiding sequences in linear time.

\begin{theorem}\label{thm:efficient-merge}
There is an algorithm that receives on input a $\pi$-avoiding sequence $S$ of length $n$ partitioned into $m$ presorted sequences $S_1, \dots, S_m$ with $m \le \frac{n}{\log n}$, and outputs $S$ sorted in $O((\log c_\pi + 1) \cdot  n)$ time even if $\pi$ is a priori unknown.
\end{theorem}

Note that \cref{thm:efficient-merge} immediately implies a divide-and-conquer algorithm for sorting $\pi$-avoiding sequences in $O_\pi(n \log^\star n)$ time.
However, that is still very far from beating the fastest previous algorithm that sorts $\pi$-avoiding permutations in $n \cdot 2^{O(|\pi|^2) + (1+o(1))\alpha(n)}$ time~\cite{ChalermsookPettieY}.

\subsection{\texorpdfstring{Prior knowledge of $\pi$}{Prior knowledge of π}}

First, let us assume that the algorithm knows (or is given) the Füredi-Hajnal limit $c_\pi$.
Moreover, we assume that the input is a permutation sequence, i.e., its all elements are pairwise distinct.

Set $d = \lceil 2 c_\pi \rceil$.
We first reduce the number of presorted sequences to $m' \le \frac{n}{d \log n}$ by merging together consecutive $d$-tuples of sequences, with potentially doing one extra two-way merge at the end.
Let $S'_1, \dots, S'_{m'}$ denote the resulting sequences.
If all sequences already got merged into a single one, the algorithm simply copies it on output and terminates.

Afterwards, the algorithm initializes a priority queue $Q$ containing sequences $S'_1, \dots, S'_{m'}$ ordered by their initial elements and proceeds in rounds.
In each round, the algorithm extracts from~$Q$ the $d+1$ sequences starting with the smallest elements.
Let $i_1, \dots, i_{d+1}$ denote their indices.
The algorithm then performs a standard $d$-way merge of $S'_{i_1}, S'_{i_2}, \dots, S'_{i_d}$ restricted to the elements smaller than the initial element (minimum) of $S'_{i_{d+1}}$.
We say that the sequences $S'_{i_1}, S'_{i_2}, \dots, S'_{i_d}$ are \emph{touched} by the algorithm in this round.
At the end of the round, each sequence $S'_{i_j}$ is pushed back in~$Q$ unless it has already been emptied.
Note that if $Q$ contains at most $d$ sequences at the beginning of any round, the algorithm simply merges them all together and terminates.
This is repeated until all sequences $S'_1, \dots, S'_{m'}$ are exhausted.
See \cref{alg:merge-known}.
Observe that the algorithm produces the same result as a standard multi-way merge and thus, it correctly outputs a sorted sequence.

\begin{algorithm}[t] 	\caption{Algorithm for merging $\pi$-avoiding input $S$ partitioned into presorted sequences $S_1, \dots, S_m$ with prior knowledge of $\pi$.}\label{alg:merge-known}
	\begin{algorithmic}
		\State $d \gets \lceil 2 c_\pi \rceil$
		\State $m' \gets \lfloor \frac{m}{d} \rfloor$
		\State $S'_1, \dots, S'_{\lceil m / d \rceil} \gets $ merges of consecutive $d$-tuples from $S_1, \dots, S_m$
		\If{$m$ is not divisible by $d$}
		\State $S'_{\lfloor m/d\rfloor} \gets$ merge of $S'_{\lfloor m/d \rfloor}$ and $S'_{\lceil m/d \rceil}$ \Comment{The number of sequences is now exactly $m'$.}
				\EndIf
		\State $Q \gets$ priority queue containing $S'_1, \dots, S'_{m'}$ ordered by their initial elements
		\While{$Q$ is non-empty}
		\If{$|Q| \le d$}
		\State{Output the merge of all sequences in $Q$ and terminate.}
		\EndIf
		\State $(i_1, i_2,  \dots, i_{d+1}) \gets$ pop $d+1$ sequences from $Q$
		\State $x \gets$ initial element of the sequence $S'_{i_{d+1}}$
		\State Output $d$-way merge of $S'_{i_1}, S'_{i_2}, \dots, S'_{i_d}$ restricted to elements smaller than $x$.
		\For{$j \in \{1, \dots, d+1\}$}
		\If{$S'_{i_j}$ is still non-empty}
		\State Push $S'_{i_j}$ back to $Q$.
		\EndIf
		\EndFor
		\EndWhile

	\end{algorithmic}
\end{algorithm}

\paragraph{Running time.}
Clearly, each element in the input incurs a time of $O(\log d)$ when written on the output, since it is obtained from a merge of the $d$ way.
It remains to bound the initial preproccessing time and the extra overhead per one round together with the total number of rounds.

The initial merging of input sequences in $d$-tuples takes $O(\log d \cdot n)$ time using the standard multi-way merge implementation with binary heap.
It suffices to implement the priority queue $Q$ using the standard binary heap.
Its initialization takes only $O(m') \subseteq O(\frac{n}{\log n})$ time.
In each round, we pop $d+1$ sequences from the priority queue~$Q$ and push at most $d+1$ sequences back for a total time of $O(d \cdot \log m') \subseteq O(d \cdot \log n)$.
Therefore, each round incurs an overhead of $O(d \cdot \log n)$.

The total running time then follows straightforwardly from the following claim.

\begin{claim}\label{claim:rounds-total}
When $S$ is a $\pi$-avoiding permutation sequence, \cref{alg:merge-known} terminates in $m' \le \frac{n}{d \cdot \log n}$ rounds.
\end{claim}

\begin{proof}
Observe that the presorted sequences $S'_1, \dots S'_{m'}$ partition $S$ into consecutive intervals with respect to the left-to-right order, while the individual rounds partition $S$ into consecutive intervals of elements with respect to their values.
Together, they form a gridding of sorts that we aim to use in conjunction with the Marcus-Tardos theorem.
Let $R$ denote the total number of rounds  and let us define an $m' \times R$ matrix $M$ as follows
\[M(i,j) = \begin{cases}
 1 &\text{if the sequence $S'_i$ was touched in the $j$th round}\\
 0 &\text{otherwise.}
\end{cases}\]

Assume for a contradiction that $R > m'$.
Since the algorithm touches exactly $d$ sequences in each round except for possibly the last one, each row in $M$ except for possibly the last one contains exactly $d$ 1-entries.
In total, that makes the number of 1-entries in $M$ is strictly larger than
\[(R-1) \cdot d \ge (R-1) \cdot 2 c_\pi \ge R \cdot c_\pi\]
where the last inequality holds since $R \ge 2$.
Now \cref{lem:marcus-tardos-rect} implies that $M$ contains an occurrence of $\pi$ because by our assumptions $\max(m', R) = R$.
Therefore, there are indices $i_1< \dots < i_{|\pi|}$ and $j_1, \dots, j_{|\pi|}$ such that $M(i_1, j_1), \dots, M(i_k,j_{|\pi|})$ are all 1-entries and $j_1, \dots , j_{|\pi|}$ is order-isomorphic to~$\pi$.
For $p \in [|\pi|]$, let $x_p$ be an arbitrary element of the sequence $S'_{i_p}$ that was outputted in the $j_p$th round.
It follows that $x_1, \dots, x_{{|\pi|}}$ form a subsequence of $S$ that is order-isomorphic to $\pi$ and we reach a contradiction.
\end{proof}

\subsection{\texorpdfstring{Without prior knowledge of $\pi$}{Without prior knowledge of π}}

Now, we consider the more general case where $\pi$ (and $c_\pi$ in particular) is not known to the algorithm.
However, we still restrict inputs to permutation sequences.
The basic idea is to guess $c_\pi$ by starting with a constant, trying the previous approach for a bounded number of rounds and doubling.

\paragraph{Description of the algorithm.}
Initially, we set $d = 1$ and the algorithm proceeds in \emph{phases} until all the elements are written on output.
At the beginning of each phase, the algorithm doubles the main parameter by setting $d \gets 2 d$, and halves the number of input sequences by merging them in pairs.
If the number of sequences at the start of the phase was odd, we perform one additional merge to ensure that we always decrease the number of sequences by at least half.
Let us denote by $d_i$ the value of $d$ in the $i$th phase.
It follows that in the $i$th phase, we work with $m_i \le \frac{m}{2^i}$ sequences and $d_i$ is equal to $2^i$.
In particular, we have $m_i \le \frac{n}{d_i \log n}$ in the $i$th phase.

Afterwards, the algorithm follows the case where $c_\pi$ was known, i.e. it proceeds in rounds where in each round, it takes the $d$-tuple of sequences with the smallest initial elements and performs a $d$-way merge until the value of the $(d+1)$-th smallest initial element is reached.
A standard binary heap is again used for efficiently retrieving these $d$ sequences.
The only change is that now the upper bound on the number of rounds is explicitly enforced rather than implicitly observed.
Specifically, the $i$th phase is terminated after $m_i$ rounds and the algorithm enters next phase.
This continues until all sequences are exhausted.
See \cref{alg:merge}.

\begin{algorithm}[t] 	\caption{Algorithm for merging $\pi$-avoiding input $S$ partitioned into presorted sequences $S_1, \dots, S_m$ without prior knowledge of $\pi$.}\label{alg:merge}
	\begin{algorithmic}
	\State $d \gets 1$
	\While{there are some elements remaining}
	\State $d \gets 2 \cdot d$
	\State $S_1, \dots, S_{\lceil m / 2 \rceil} \gets $ merges of consecutive pairs from $S_1, \dots, S_m$
	\If{$m$ is odd}
	\State $S_{\lfloor m/2\rfloor} \gets$ merge of $S_{\lfloor m/2 \rfloor}$ and $S_{\lceil m/2 \rceil}$ \Comment{The number of sequences is at most $\frac{m}{2^i}$.}
	\EndIf
	\State $m \gets \lfloor \tfrac{m}{2} \rfloor$ 
	\State $Q \gets$ priority queue with all non-empty sequences ordered by their initial elements
	\State $r \gets 1$ \Comment{Counter for the number of rounds}
	\While{$Q$ is non-empty and $r \le m$}
	\If{$|Q| \le d$}
	\State{Output the merge of all sequences in $Q$ and terminate.}
	\EndIf
	\State $(i_1, i_2,  \dots, i_{d+1}) \gets$ pop $d+1$ sequences from $Q$
	\State $x \gets$ initial element of the sequence $S_{i_{d+1}}$
	\State Output $d$-way merge of $S_{i_1}, S_{i_2}, \dots, S_{i_d}$ restricted to elements smaller than $x$.	
	\For{$j \in \{1, \dots, d+1\}$}
	\If{$S_{i_j}$ is still non-empty}
	\State Push $S_{i_j}$ back to $Q$.
	\EndIf
	\EndFor
		\State $r \gets r +1$
	\EndWhile
	\EndWhile
	\end{algorithmic}
\end{algorithm}

It is important that whenever a new phase begins, the algorithm does not reset to the very beginning as there is no reason to do so.
The output to this point is a valid prefix of the sorted sequence.
Thus, the next phase continues with only the remaining elements in each sequence.

As before, the algorithm clearly outputs a correctly sorted sequence, and it only remains to bound its running time.

\begin{claim}\label{claim:phases-total}
On a $\pi$-avoiding permutation sequence $S$, \cref{alg:merge} terminates after at most $\lceil \log c_\pi \rceil + 1$ phases.
\end{claim}

\begin{proof}
Assume for a contradiction that more phases than $\lceil \log c_\pi\rceil + 1$ are needed and set $i = \lceil \log c_\pi\rceil + 1 $.
In other words, the algorithm is not able to terminate in the $i$th phase. 
Observe that $ 2^i \ge 2 c_\pi$ and therefore, we have $d_i \ge 2 c_\pi$.
Recall that $m_i$ denotes the number of sequences in the $i$th phase. 
But since $d_i \ge 2 c_\pi$ and $m_i \le \frac{n}{d_i \log n}$, the exact same argument as in \cref{claim:rounds-total} implies that $m_i$ rounds suffice to merge all remaining elements, and thus the algorithm should have terminated in the $i$th phase regardless of the number of elements processed in previous phases. 
\end{proof}

\paragraph{Running time.}
As before, we charge each element with the time needed to retrieve it when it is outputted.
In the $i$th phase, the algorithm outputs elements using $2^i$-way merges, and thus each element incurs a cost of $O(i)$ which is at most $O(\log c_\pi + 1)$ by \cref{claim:phases-total}.

Now, we show that we only use linear extra overhead in each phase, again totaling $O((\log c_\pi + 1) \cdot n)$ time by \cref{claim:phases-total}.
The merging at the beginning of each phase takes clearly only linear time, since each element participates in at most two 2-way merges.
Moreover, the $i$th phase consists of at most $m_i \le \frac{n}{2^i \log n}$ rounds by design.
Each round can be handled with an overhead of $O(2^i \cdot \log m_i) \subseteq O(2^i \cdot \log n)$ for retrieving the $2^i$ smallest sequences from the binary heap at the beginning of the round and inserting them back at the end of the round.
In total, this produces an overhead of $O(2^i \cdot \log n \cdot m_i) \subseteq O(n)$ per each phase as promised.

\subsection{General sequences}
\label{sec:merge-general}
Finally, we remove the restriction of inputs to permutation sequences.
Let us first describe a standard way used to deal with general pattern-avoiding sequences in previous works~\cite{FOCS15,BerensohnKO2023}.
The repeated elements in an arbitrary $\pi$-avoiding sequence $S$ can be perturbed to obtain a permutation sequence $S'$ that avoids some pattern $\pi'$ where $|\pi'| = 2 |\pi|$.
It follows that a linear sorting algorithm for pattern-avoiding permutations can be applied to $S'$ in order to sort $S$.
See~\cite{FOCS15} for more details.
However, the relationship between $c_\pi$ and $c_{\pi'}$ is not known in general, and as a result, we might lose the optimal asymptotic dependence on the pattern $\pi$.

For this reason, we take a different path of carefully analyzing the behavior of \cref{alg:merge} on general $\pi$-avoiding sequences.
Specifically, we launch \cref{alg:merge} on arbitrary $\pi$-avoiding input while breaking all ties according to their position on input, i.e., we perform a stable merge.
Note that this tie-breaking can occur both when selecting the $d$-tuple of the smallest sequences in each round as well as during any merging.
We claim that the asymptotic bound on runtime does not deteriorate.

Observe that \cref{claim:phases-total} is the only place in the analysis of \cref{alg:merge} where we used the fact that $S$ is a permutation sequence.
It therefore suffices to show that analogous claim holds for general sequences and the promised runtime follows immediately.

\begin{lemma}\label{lem:phases-total-general}
On an arbitrary $\pi$-avoiding sequence $S$, \cref{alg:merge} terminates after at most $\lceil \log c_\pi\rceil + 2$  phases.
\end{lemma}
\begin{proof}
Assume for a contradiction that more phases than $\lceil \log c_\pi\rceil + 2$ are needed and set $i = \lceil \log c_\pi\rceil + 2$.
Our goal is to show that, in fact, the algorithm should have terminated in the $i$th phase.
Our choice of $i$ guarantees that $2^i \ge  4 c_\pi$ and therefore we have $d_i \ge 4 c_\pi$.

We distinguish between two different types of rounds during the execution of \cref{alg:merge} on general sequences.
A round is \emph{heavy} if all elements processed during it are equal in value; otherwise we say that the round is \emph{light}.
A \emph{value} of a heavy round is simply the singular value of all the elements processed in that round.
We shall bound the number of heavy and light rounds independently.
Furthermore, we focus only on rounds that touch exactly $d$ sequences.
Let us call such rounds \emph{full}.
Observe that there can be at most one non-full round, the final one, and we can handle this round separately regardless of its type.

First, let us focus on the full heavy rounds.
Observe that if both the $i$th and $j$th rounds are heavy and of the same value, then they necessarily touch a completely disjoint set of sequences.
Moreover, if $i < j$ then all the sequences touched in the $i$th round precede those touched in the $j$th round in the left-to-right order given by their indices.
Let $H$ denote the number of full heavy rounds and let $x_1 < x_2 < \dots < x_p$ be all the distinct values of some full heavy round where $p \le H$.
We claim that $H \le \frac{m_i}{4}$.
Suppose not.
We define an $m_i\times p$ matrix $M$ as follows
\[M(j,k) = \begin{cases}
	1 &\text{if the sequence $S'_j$ was touched by a full heavy round of value $x_k$}\\
	0 &\text{otherwise.}
\end{cases}\]
As we already argued, no sequence was touched by two different heavy rounds of the same value.
Moreover, we only count full rounds that touched exactly $d$ sequences and therefore, the number of 1-entries in $M$ is 
\[H \cdot d_i > \frac{m_i}{4} \cdot 4c_\pi = m_i \cdot c_\pi.\]

It follows that $M$ again contains an occurrence of $\pi$ by \cref{lem:marcus-tardos-rect} and since all its elements come from pairwise distinct sequences, we can find an occurrence of $\pi$ in the original input sequence $S$.
Therefore, we reached a contradiction and conclude that $H \le \frac{m_i}{4}$.

We would like to bound the number of light rounds in a similar fashion to the proof of \cref{claim:rounds-total}.
However, two elements processed in consecutive light rounds might have the same value, and thus, the gridding induced by sequences and rounds would be ill-formed.
In particular, an occurrence of~$\pi$ in this gridding would not necessarily imply an occurrence of~$\pi$ in the input sequence.
On the other hand, notice that two light rounds separated by another light round must involve elements of completely disjoint values.
Therefore, the solution is to consider only every other light round.

Let $L$ denote the number of full light rounds.
We claim that $L \le \frac{m_i}{2}$.
Assume for a contradiction that $L > \frac{m_i}{2}$ and let $\ell_1, \dots, \ell_{\lceil L/2\rceil}$ be the indices of odd full light rounds.
We construct an $m_i \times \lceil \frac{L}{2}\rceil$ matrix $M$ as follows
\[M(j,k) = \begin{cases}
	1 &\text{if the sequence $S'_j$ was touched in the ($\ell_k$)th round}\\
	0 &\text{otherwise.}
\end{cases}\]

Each full round touches exactly $d_i$ sequences, and thus the number of 1-entries in $M$ is
\[\left\lceil \tfrac{L}{2} \right\rceil \cdot d_i > \frac{m_i}{4} \cdot 4 c_\pi = m_i \cdot c_\pi.\]
Again, it follows that $M$ contains an occurrence of $\pi$ due to \cref{lem:marcus-tardos-rect}.
We argue that it induces an occurrence of $\pi$ in the original input sequence $S$.
We recover a subsequence $x_1, \dots, x_{|\pi|}$ of $S$ from the occurrence of $\pi$ in $M$ exactly as in the proof of \cref{claim:rounds-total}.
In this general case, we must also verify that all the elements are distinct in value.
But crucially, this follows by our consideration of only every other light round and hence, $x_1, \dots, x_{|\pi|}$ is order-isomorphic to $\pi$ and we reach a contradiction.
Thus, we have $L \le \frac{m_i}{2}$.

Putting it all together, the $i$th phase consisted of at most $\frac{3}{4} m_i$ full rounds, and accounting for the one additional final round makes the total number of rounds at most $\frac{3}{4}m_i + 1$.
Furthermore, the number of rounds must be an integer and thus we know that at most $\lfloor \frac{3}{4}m_i \rfloor + 1$ rounds occurred. 
But it is easily verified that $\lfloor \frac{3}{4}m_i \rfloor + 1 \le m_i$ for every integer $m_i \ge 1$ and it follows that \cref{alg:merge} should have terminated in the $i$th phase.
\end{proof}

\section{Sorting many short sequences}
\label{sec:sorting-short}

In this section, we show how to use decision trees to efficiently sort a large set of short sequences.
Step by step, we will construct an algorithm summarized by the following theorem.

\begin{theorem}\label{thm:sort-many}
	Let $S_1, \ldots, S_{n/k}$ be $\pi$-avoiding sequences, each of length $k$. Then there is an algorithm that sorts all $S_i$'s in time $O((\log c_\pi+1) \cdot n) + 2^{2^{O(k \log k)}}$ even if $\pi$ is a priori unknown.
\end{theorem}

\subsection{Decision trees}

A \emph{decision tree} $T$ for sorting a sequence of length $k$ is formally a full binary tree with each internal node labeled by a pair $(i,j) \in [k]^2$ and each leaf labeled by a permutation of length $k$.
For a sequence $S= (s_1, \dots, s_k)$ of length $k$, the \emph{run of $T$ on $S$} is the root-to-leaf path in $T$ such that in an internal node $v$ labeled with $(i,j)$, the path continues to the left subtree of $v$ if and only if $s_i \le s_j$.
For a sequence $S$ of length $k$, the \emph{result of a run of $T$ on $S$} is the label of the leaf reached by the run of $T$ on $S$.
We say that a decision tree $T$ \emph{sorts} a sequence $S = s_1, \dots, s_k$ if the result of the run of $T$ on $S$ is a permutation $\sigma$ such that $s_{\sigma_1} \le \dots \le s_{\sigma_k}$.
See \Cref{fig:decision-tree}.
Observe that for a permutation $\sigma$, a decision tree $T$ sorts $\sigma$ precisely if the result of applying $T$ on $\sigma$ is its inverse $\sigma^{-1}$.

\begin{figure}
		\centering
\forestset{
	default preamble={
		for tree={draw,ellipse,s sep=5mm, l sep=0mm}
	}
}
		\begin{forest}
			[{$(1,2)$}	
				[{$(2,3)$} 
					[{$(2,3)$}[{$1, 2, 3$},rectangle] [{$1, 2, 3$},rectangle] ] 
					[{$(1,3)$} [{$1, 3, 2$},rectangle] [{$3, 1, 2$},rectangle] ]]
				[{$(1,3)$},edge={dashed,red,line width=1.5pt}
					[{$(1,3)$} [{$2, 1, 3$},rectangle] [{$2, 1, 3$},rectangle]] 
					[{$(2,3)$},edge={dashed,red,line width=1.5pt} [{$2, 3, 1$},rectangle,edge={dashed,red,line width=1.5pt}] [{$3, 2, 1$},rectangle] ]]
			]
		\end{forest}
\caption{A decision tree $T$ for sorting a sequence of length~$3$. A run on input sequence $(s_1, s_2, s_3) = (7,2,3)$ is highlighted and it reaches a leaf labeled by the permutation $2, 3, 1$. Indeed, we have $s_2 \le s_3 \le s_1$ and $T$ correctly sorts $S$.
Notice that the tree contains many unnecessary branches since we require the decision trees to be full binary trees.}\label{fig:decision-tree}
\end{figure}

As our main tool, we make use of the fact that for an arbitrary set of permutations there exists a decision tree whose height matches the information-theoretic lower bound up to an extra linear term.

\begin{theorem}[Fredman~\cite{Fredman}]\label{thm:fredman}
For every set ${\mit\Gamma}$ of permutations of length $n$, there exists a decision tree of height at most $\log |{\mit\Gamma}| + 2n$ that correctly sorts every $\pi \in {\mit\Gamma}$.
\end{theorem}

It follows by the Marcus-Tardos theorem~\cite{MarcusTardos} and the results of Cibulka~\cite{Cibulka2009} that the number of $\pi$-avoiding permutations of length~$n$ is bounded by $c_\pi^{O(n)}$.
Therefore, there is a decision tree of height $O((\log c_\pi+1) \cdot  n)$ that correctly sorts every $\pi$-avoiding permutation of length $n$.

%

We construct an algorithm that receives several sequences of length $k$ on input and sorts each of them asymptotically in the same time as given by \cref{thm:fredman}, without explicit knowledge of ${\mit\Gamma}$.
The only caveat is that it needs an additional overhead that is roughly doubly exponential in $k$.

\begin{lemma}\label{lem:many-sort}
	Let $S_1, \ldots, S_{m}$ be a set of permutation sequences, each of length $k$, and let ${\mit\Gamma}$ be the set of all permutations order-isomorphic to some sequence $S_i$. \cref{alg:sorting-many} sorts all $S_i$'s in time $O((\log |{\mit\Gamma}| + k) \cdot m ) + 2^{2^{O(k \log k)}}$.
\end{lemma}

We first provide an upper bound on the total number of decision trees of given height for sorting a sequence of length $k$.
Note that we are neglecting here any optimization of the lower-order terms.

\begin{lemma}\label{lem:decision-trees-count}
The number of decision trees of height at most $h$ for sorting a sequence of length $k$ is at most $2^{2^{h + o(k) +o(h)}}$.
\end{lemma}
\begin{proof}
Let us first count the number of such decision trees of height exactly $h$.
Since the shape of the tree is given, the decision tree is completely determined by its labels.
There are $2^h-1 \le 2^h$ internal nodes, each labeled by one of $k^2$ possible ordered pairs.
Additionally, there are exactly $2^h$ leaves, each marked with a permutation of length $k$.
Together, that makes the number of such decision trees of height $h$ at most
\[(k^2)^{2^h} \cdot (k!)^{2^h} \le 2^{2 \cdot \log k \cdot 2^h} \cdot 2 ^{k \cdot \log k \cdot 2^h} \le 2^{2^{h + o(k)}}.\]

We obtain the promised bound by summing over all heights up to $h$ 
\[h \cdot 2^{2^{h + o(k)}} \le 2^{2^{h + o(k) + o(h)}}.\qedhere\]
\end{proof}

Let us take a brief detour and sketch how a prior knowledge of $\pi$ can be used to efficiently sort many $\pi$-avoiding sequences of length $k$.
Knowing $\pi$,  we first compute the set of all $\pi$-avoiding permutations of length $k$ by enumerating all permutations of length $k$ and testing for pattern-avoidance.
Afterwards, we enumerate all the decision trees of height $O((\log c_\pi+1) \cdot k)$ and for each of them, we test whether it successfully sorts all the precomputed $\pi$-avoiding permutations of length~$k$.
All of this can be done in $2^{2^{O((\log c_\pi+1) \cdot k)}}$ time.
Once we find a suitable decision tree guaranteed by \cref{thm:fredman}, we simply apply it on every input sequence in $O((\log c_\pi+1) \cdot k)$ time for a total of $O((\log c_\pi +1 ) \cdot n)$.
Note that this approach follows closely the use of decision trees in the celebrated asymptotically optimal algorithm for minimum spanning tree by Pettie and Ramachandran~\cite{PettieR02}.

\medskip
However, we propose a more elegant solution that does not require prior knowledge of $\pi$ and additionally works for sets of permutations defined through different means than pattern avoidance.
Rather than computing the optimal decision tree upfront, we work with a candidate tree instead.
For each sequence $S_i$ on input, we try applying the current decision tree $T$ on $S_i$ and test in $O(k)$ time whether $T$ sorts $S_i$.
If that is the case, we move on to the next sequence with the same decision tree.
Otherwise, we throw away the current decision tree, replace it with a different one, and try sorting $S_i$ again.

\paragraph{Description of the algorithm.}
Now, we describe the algorithm in full detail.
Recall that we are given $m$ sequences $S_1, \ldots, S_{m}$ of length $k$ on input.
For now, assume that each $S_i$ is a permutation sequence.

Let $T_1, T_2, \ldots$ be an arbitrary enumeration of all the decision trees for sorting a sequence of length $k$ in the increasing order given by their heights, starting with the decision trees of height exactly $k$.
For now, we just assume that we can efficiently enumerate them in this order, leaving the implementation details for later.
The algorithm keeps a current candidate for the optimal decision tree $T$.
At the beginning, we set $T \gets T_1$.

Now, for each $S_i$, we do the following.
The algorithm computes the result $\sigma$ of the run of the tree~$T$ on~$S_i$.
If $T$ sorts $S_i$, i.e. we have $s_{\sigma_1} \le  \dots \le s_{\sigma_k}$, the algorithm accordingly rearranges $S_i$ and moves on to the next sequence on input.
Otherwise, the algorithm changes $T$ to the next decision tree in the enumeration sequence and repeats the previous step.
See \cref{alg:sorting-many}.

\begin{algorithm}[t] 	\caption{Algorithm for sorting $\pi$-avoiding sequences $S_1, \dots, S_m$ of length $k$.}\label{alg:sorting-many}
	\begin{algorithmic}
		\State $T \gets$ first decision tree in the enumeration sequence 
		\For{$i \in \{1, \dots, m\}$}
		\State $\sigma \gets$ the result of applying $T$ on $S_i = s^i_1, \dots, s^i_k$ 
		\While{$s^i_{\sigma_1}, \dots, s^i_{\sigma_k}$ is not sorted}
		\State $T \gets$ next decision tree in the enumeration sequence
		\State $\sigma \gets$ the result of the run of $T$ on $S_i$ 		
		\EndWhile
		\State $S_i \gets s^i_{\sigma_1}, \dots, s^i_{\sigma_k}$ \Comment{Sorting $S_i$.}
		\EndFor
	\end{algorithmic}
\end{algorithm}

\paragraph{}
If the algorithm terminates, it clearly sorts all the input sequences, since it moves to the next sequence only after it verifies that the current sequence is sorted.
Moreover, all sequences of length $k$ can be sorted by a decision tree of height $O(k \log k)$, and thus the algorithm must eventually reach a decision tree capable of sorting all input.
It follows that the algorithm always terminates and correctly sorts all the sequences on input.
It remains to bound its running time.

\paragraph{Implementation details.}
Let us show how to efficiently implement the enumeration of decision trees.
A decision tree of height $h$ is determined by $2^h-1$ ordered pairs from $[k]^2$ and $2^h$ permutations of length $k$.
At the beginning of the algorithm, we enumerate all the permutations of length $k$ and store them in a separate array $P$.
A decision tree can then be represented by an array of $2^h-1$ values from the set $[k^2]$, interpreted as ordered pairs, and $2^h$ values from the set $[k!]$, interpreted as indices into the array $P$.
Moreover, a single traversal of the tree determined by an input sequence can easily be performed in $O(h)$ time over this representation.
It suffices to arrange the internal nodes in the same way as in the standard implementation of a binary heap, i.e., a node at index $i$ has its children located at indices $2i$ and $2i+1$.

All decision trees of height $h$ can now easily be enumerated in lexicographic order with respect to this representation.
Moreover, standard analysis shows that such enumeration can be implemented with $O(1)$ amortized time per one step.
It remains to handle the situation when all decision trees of height $h$ have already been enumerated and we should start enumerating the decision trees of height $h+1$.
However, this happens so infrequently that we can reset and prepare the representation for decision trees of height $h+1$ from scratch without affecting the constant amortized time per one step.

\paragraph{Running time.}
\cref{thm:fredman} guarantees the existence of a decision tree of height $\log |{\mit\Gamma}| + 2k$ that can sort all sequences on input.
Therefore, if \cref{alg:sorting-many} reaches this decision tree in its enumeration, it does not iterate any further and all subsequent sequences are sorted on the first attempt.
In particular, the height of the decision tree $T$ in \cref{alg:sorting-many} is at most $\log |{\mit\Gamma}| + 2k$ throughout its whole runtime.

We say that a run of $T$ on $S_i$ in \cref{alg:sorting-many} is \emph{successful} if $T$ sorts $S_i$ and otherwise, it is \emph{unsuccessful}.
Observe that each input sequence $S_i$ incurs exactly one successful run and arbitrarily many unsuccessful ones.
Therefore, the number of successful runs is exactly $m$, making the total time taken by them exactly $O((\log |{\mit\Gamma}| + k) \cdot m)$.
On the other hand, each unsuccessful run causes \cref{alg:sorting-many} to advance to the next decision tree and thus, the number of unsuccessful runs is bounded by the number of decision trees considered.
By \cref{lem:decision-trees-count}, this is at most $2^{2^{O(k \log k)}}$ when we crudely bound $h$ to be at most $O(k \log k)$.
Therefore, the total time spent on unsuccessful runs is at most $2^{2^{O(k \log k)}} \cdot h \subseteq 2^{2^{O(k \log k)}}$.
Putting these together, we see that \cref{alg:sorting-many} terminates in  $O((\log |{\mit\Gamma}| + k)\cdot m ) + 2^{2^{O(k \log k)}}$ time, proving \cref{lem:many-sort}.
Consequently, if all the input sequences are $\pi$ avoiding, then $|{\mit\Gamma}| \leq c_\pi^{O(k)}$ and \cref{alg:sorting-many} terminates in $O((\log c_\pi + 1)\cdot n) + 2^{2^{O(k \log k)}}$ time.
This wraps up the proof of \cref{thm:sort-many} for permutation sequences.

\subsection{General sequences}

As in \cref{sec:merge-general}, we could deal with general sequences by perturbing and invoking the same scheme for a larger avoided pattern $\pi'$.
However as before, we avoid this approach in order to retain the asymptotically optimal dependence on $\pi$.

In fact, we will repeat our previous approach and simply run \cref{alg:sorting-many} while breaking all ties according to their position on input, i.e., performing a stable sort.
Equivalently, we replace each sequence $s_1, \dots, s_k$ with a sequence $s'_1, \dots, s'_k$ where $s'_i = s_i + i\cdot \epsilon$ for a sufficiently small $\epsilon > 0$.
Observe that in this way order-isomorphism is preserved, i.e., two order-isomorphic sequences remain order-isomorphic even after this modification.

The only thing we need to argue, is that the number of non-order-isomorphic $\pi$-avoiding sequences of length $k$ is bounded in similar fashion to the number of different $\pi$-avoiding permutations.
The bound on runtime then follows, since \cref{lem:many-sort} is completely independent of pattern-avoidance and its bound depends only on the number of non-order-isomorphic sequences on input.
Recall that we can uniquely associate a unique $\pi$-avoiding $n \times n$ permutation matrix to each $\pi$-avoiding permutation sequence of length $k$.
Then we can use the upper bound on the number of $\pi$-avoiding permutation matrices by Cibulka~\cite{Cibulka2009} to see that the number of non-order-isomorphic $\pi$-avoiding permutation sequences of length $k$ is at most $(2.88 c_\pi^2)^n$.

Similarly, each general $\pi$-avoiding sequence corresponds to a unique $\pi$-avoiding  $n \times m$ matrix where $m \le n$, each column contains exactly one 1-entry and each row contains at least one 1-entry.
We show that the number of such matrices is bounded in a similar way.
In fact, we prove a stronger claim by bounding the number of all $\pi$-avoiding $n \times n$ matrices with exactly $n$ 1-entries. 
Let $T_\pi(m,n)$ denote the set of all $\pi$-avoiding  $m \times m$ matrices with exactly $n$ 1-entries.

Cibulka~\cite{Cibulka2009} claims that arguments similar to his proof relating together the Stanley-Wilf and Füredi-Hajnal limits can be used to show that $\sqrt[n]{|T_\pi(n,n)|} \in O(c_\pi^2)$.
Unfortunately, he does not provide any further justification.
Therefore, we include a self-contained proof of this fact.

\begin{theorem}\label{thm:furedi-stanley}
For every pattern $\pi$ and every $n \in \mathbb{N}$, we have $|T_\pi(n,n)| \le (240 c_\pi^2)^n$.
\end{theorem}
\begin{proof}
We assume that $c_\pi \ge 1$, since otherwise $\pi$ is the singleton permutation and $T_\pi(n,n) = \emptyset$ for every $n \ge 1$.

The first part of the argument follows the reduction of Klazar~\cite{Klazar}.
Let $T_\pi(n)$ denote the set of all $\pi$-avoiding $n \times n$ matrices, regardless of the number of 1-entries. 
We show by induction that for every~$i$, we have $|T_\pi(2^i)| \le 15^{c_\pi \cdot 2^i}$.
This clearly holds for $i = 0$, so let $i \ge 1$.
Observe that every matrix $M \in T_\pi(2^i)$ can be obtained from some matrix $ M' \in T_\pi(2^{i-1})$ by replacing each entry $\alpha$ with a $2\times 2$ block that is all-zero if and only if $\alpha = 0$.
Moreover, by the Marcus-Tardos theorem\cite{MarcusTardos} there are at most $c_\pi \cdot 2^{i-1}$ 1-entries in $M'$ and for each of them, there are exactly 15 possible $2 \times 2$ blocks to choose from.
Thus, we get
\[|T_\pi(2^i)| \le 15^{c_\pi \cdot 2^{i-1}} \cdot |T_\pi(2^{i-1})| \le 15^{c_\pi \cdot 2^{i-1}} \cdot 15^{c_\pi \cdot 2^{i-1}} \le 15^{c_\pi \cdot 2^{i}} \]

Set $a = \lfloor \log (n/ c_\pi) \rfloor$ and observe that $\frac{n}{2 c_\pi} \le 2^a \le \frac{n}{c_\pi}$.
In order to simplify the arguments, we bound $|T_\pi(m, n)|$ where $m = 2^{a+1} \cdot c_\pi$.
The desired bound on $|T_\pi(n,n)|$ follows since $m \ge n$ and $|T_\pi(\cdot,n)|$ is a non-decreasing function.
Similarly to before, every matrix $M \in T_\pi(m,n)$ can be obtained from some $M' \in T_\pi(2^a)$ by replacing each entry $\alpha$ with a $2c_\pi \times 2c_\pi$ block that is all-zero if and only if $\alpha = 0$.

We encode $M$ from $M'$ in two steps.
First, we distribute all $n$ 1-entries of $M$ among the blocks obtained from the 1-entries of $M'$.
This is the standard problem of counting the number of ways to distribute indistinguishable balls into bins.
By the Marcus-Tardos theorem\cite{MarcusTardos}, the matrix $M'$ contains at most $c_\pi \cdot 2^a \le n$ 1-entries.
Therefore, the number of ways to distribute the 1-entries of $M$ into the blocks obtained from the 1-entries of $M'$ is at most $\binom{2n-1}{n} \le 4^n$.

At this point, we know how the $n$ 1-entries of $M$ are distributed among the $2c_\pi \times 2c_\pi$ blocks.
It remains to encode the specific location of each 1-entry within its assigned block.
There are $(2c_\pi)^2$ possibilities for each 1-entry, making a total of $(2c_\pi)^{2n}$.
Putting everything together, we obtain
\[|T_\pi(n,n)| \le |T_\pi(m,n)| \le |T_\pi(2^a)| \cdot 4^n \cdot (2c_\pi)^{2n} \le 15^{c_\pi \cdot \frac{n}{c_\pi}} \cdot 16^n (c_\pi^2)^n \le (240 c_\pi^2)^n. \qedhere\]
\end{proof}

This wraps up the proof of Theorem~\ref{thm:sort-many}.

\section{The algorithm}\label{sec:algo}

We can finally assemble the asymptotically optimal algorithm by combining the merging algorithm obtained in \cref{sec:merge} and the algorithm for sorting many short sequences obtained in \cref{sec:sorting-short}.

\begin{proof}[Proof of \cref{thm:main-result}]
We can assume that $n$ is sufficiently large as otherwise, we can sort the input using any standard comparison-based sorting algorithm without affecting the asymptotic runtime.

Set $k = \lceil \log^{(3)}n \rceil$ where $\log^{(t)}n$ denotes the logarithm base $2$ iterated $t$ times.
Now, we split $S$ into $\lceil n/k \rceil$ sequences $S_1, \dots, S_{\lceil n/k\rceil}$ of $k$ elements, except possibly $S_{\lceil n/k\rceil}$, which can be shorter.
We sort $S_{1}, \dots, S_{\lfloor n/k\rfloor}$ using \cref{alg:sorting-many} and additionally, we sort $S_{\lceil n/k\rceil}$ using any $O(k \log k)$ comparison-based sorting algorithm.
By \cref{thm:sort-many}, this takes $O((\log c_\pi+1) \cdot n) + 2^{2^{O(k \log k)}}$ time.
Plugging in $k$, we see that $2^{2^{O(k \log k)}} \subseteq o(n)$ and thus, this step takes only $O((\log c_\pi+1) \cdot n)$ time.\footnote{It would in fact suffice to set $k = C \cdot\frac{\log^{(2)} n}{\log^{(3)}n}$ for a suitably small constant $C$.}

Now it suffices to perform a bottom-up merge sort using \cref{alg:merge}  and starting with $S_1, \dots, S_{\lceil n/k\rceil}$.
To be precise, we first merge $\frac{\log^{(2)}n}{\log^{(3)}n}$-tuples of consecutive sequences in order to reduce the number of sequences to at most $\frac{n}{\log^{(2)}n}$.
Afterwards, we repeat the same step with $\frac{\log n}{\log^{(2)}n}$-tuples reducing the number of sequences to at most $\frac{n}{\log n}$.
Finally, all remaining sequences can be merged by a single run of \cref{alg:merge}.
These can be seen as three layers of recursive merge sort where each layer takes $O((\log c_\pi+1) \cdot n)$ time by \cref{thm:efficient-merge}.
Therefore, the merging as a whole finishes in $O((\log c_\pi+1) \cdot n)$ time and the same bound applies to the whole algorithm.
\end{proof}

%
%

\small
\bibliographystyle{alphaurl}
\bibliography{main}

\newcommand{\etalchar}[1]{$^{#1}$}
\begin{thebibliography}{BKTW22}

\bibitem[Arr99]{Arratia99}
Richard Arratia.
\newblock On the {S}tanley-{W}ilf conjecture for the number of permutations
  avoiding a given pattern.
\newblock {\em Electron. J. Comb.}, 6, 1999.
\newblock \href {https://doi.org/10.37236/1477} {\path{doi:10.37236/1477}}.

\bibitem[Art07]{Arthur}
David Arthur.
\newblock Fast sorting and pattern-avoiding permutations.
\newblock In {\em 4th Workshop on Analytic Algorithmics and Combinatorics,
  {ANALCO} 2007}, pages 169--174. SIAM, 2007.
\newblock \href {https://doi.org/10.1137/1.9781611972979.1}
  {\path{doi:10.1137/1.9781611972979.1}}.

\bibitem[BBGT24]{BonnetBourneufEtAl2024}
Edouard Bonnet, Romain Bourneuf, Colin Geniet, and Stéphan Thomassé.
\newblock Factoring pattern-free permutations into separable ones.
\newblock In {\em ACM-SIAM Symposium on Discrete Algorithms, {SODA} 2024},
  pages 752--779, 2024.
\newblock \href {https://doi.org/10.1137/1.9781611977912.30}
  {\path{doi:10.1137/1.9781611977912.30}}.

\bibitem[BKO24]{BerensohnKO2023}
Benjamin~Aram Berendsohn, L{\'{a}}szl{\'{o}} Kozma, and Michal Opler.
\newblock Optimization with pattern-avoiding input.
\newblock In {\em 56th Annual {ACM} Symposium on Theory of Computing, {STOC}
  2024}, pages 671--682. {ACM}, 2024.
\newblock \href {https://doi.org/10.1145/3618260.3649631}
  {\path{doi:10.1145/3618260.3649631}}.

\bibitem[BKTW22]{BonnetKTW22}
{\'{E}}douard Bonnet, Eun~Jung Kim, St{\'{e}}phan Thomass{\'{e}}, and
  R{\'{e}}mi Watrigant.
\newblock Twin-width {I:} tractable {FO} model checking.
\newblock {\em J. {ACM}}, 69(1):3:1--3:46, 2022.
\newblock \href {https://doi.org/10.1145/3486655} {\path{doi:10.1145/3486655}}.

\bibitem[CGJ{\etalchar{+}}23]{BSTdecomp}
Parinya Chalermsook, Manoj Gupta, Wanchote Jiamjitrak, Nidia~Obscura Acosta,
  Akash Pareek, and Sorrachai Yingchareonthawornchai.
\newblock Improved pattern-avoidance bounds for greedy bsts via matrix
  decomposition.
\newblock In {\em {ACM-SIAM} Symposium on Discrete Algorithms, {SODA} 2023},
  pages 509--534. {SIAM}, 2023.
\newblock \href {https://doi.org/10.1137/1.9781611977554.ch22}
  {\path{doi:10.1137/1.9781611977554.ch22}}.

\bibitem[CGK{\etalchar{+}}15]{FOCS15}
Parinya Chalermsook, Mayank Goswami, L{\'{a}}szl{\'{o}} Kozma, Kurt Mehlhorn,
  and Thatchaphol Saranurak.
\newblock Pattern-avoiding access in binary search trees.
\newblock In {\em {IEEE} 56th Annual Symposium on Foundations of Computer
  Science, {FOCS} 2015}, pages 410--423. {IEEE} Computer Society, 2015.
\newblock \href {https://doi.org/10.1109/FOCS.2015.32}
  {\path{doi:10.1109/FOCS.2015.32}}.

\bibitem[Cib09]{Cibulka2009}
Josef Cibulka.
\newblock {On constants in the Füredi–Hajnal and the Stanley–Wilf
  conjecture}.
\newblock {\em Journal of Combinatorial Theory, Series A}, 116(2):290--302,
  2009.
\newblock \href {https://doi.org/10.1016/j.jcta.2008.06.003}
  {\path{doi:10.1016/j.jcta.2008.06.003}}.

\bibitem[CLRS22]{Cormen}
T.H. Cormen, C.E. Leiserson, R.L. Rivest, and C.~Stein.
\newblock {\em Introduction to Algorithms, fourth edition}.
\newblock MIT Press, 2022.
\newblock URL: \url{https://books.google.de/books?id=HOJyzgEACAAJ}.

\bibitem[CPY24]{ChalermsookPettieY}
Parinya Chalermsook, Seth Pettie, and Sorrachai Yingchareonthawornchai.
\newblock Sorting pattern-avoiding permutations via 0-1 matrices forbidding
  product patterns.
\newblock In {\em {ACM-SIAM} Symposium on Discrete Algorithms, {SODA} 2024},
  pages 133--149. {SIAM}, 2024.
\newblock \href {https://doi.org/10.1137/1.9781611977912.7}
  {\path{doi:10.1137/1.9781611977912.7}}.

\bibitem[DHI{\etalchar{+}}09]{DHIKP}
Erik~D. Demaine, Dion Harmon, John Iacono, Daniel~M. Kane, and Mihai
  P{\u{a}}tra{\c{s}}cu.
\newblock The geometry of binary search trees.
\newblock In {\em {ACM-SIAM} Symposium on Discrete Algorithms, {SODA} 2009},
  pages 496--505. {SIAM}, 2009.
\newblock \href {https://doi.org/10.1137/1.9781611973068.55}
  {\path{doi:10.1137/1.9781611973068.55}}.

\bibitem[FH92]{FurediHajnal}
Zolt{\'{a}}n F{\"{u}}redi and P{\'{e}}ter Hajnal.
\newblock {Davenport-Schinzel theory of matrices}.
\newblock {\em Discret. Math.}, 103(3):233--251, 1992.
\newblock \href {https://doi.org/10.1016/0012-365X(92)90316-8}
  {\path{doi:10.1016/0012-365X(92)90316-8}}.

\bibitem[Fox13]{jfox}
Jacob Fox.
\newblock Stanley-wilf limits are typically exponential, 2013.
\newblock \href {http://arxiv.org/abs/1310.8378} {\path{arXiv:1310.8378}}.

\bibitem[Fre76]{Fredman}
Michael~L. Fredman.
\newblock How good is the information theory bound in sorting?
\newblock {\em Theor. Comput. Sci.}, 1(4):355--361, 1976.
\newblock \href {https://doi.org/10.1016/0304-3975(76)90078-5}
  {\path{doi:10.1016/0304-3975(76)90078-5}}.

\bibitem[GG19]{GoyalGupta}
Navin Goyal and Manoj Gupta.
\newblock Better analysis of {GREEDY} binary search tree on decomposable
  sequences.
\newblock {\em Theoretical Computer Science}, 776:19--42, 2019.
\newblock \href {https://doi.org/10.1016/j.tcs.2018.12.021}
  {\path{doi:10.1016/j.tcs.2018.12.021}}.

\bibitem[GM14]{GM_PPM}
Sylvain Guillemot and D{\'{a}}niel Marx.
\newblock Finding small patterns in permutations in linear time.
\newblock In {\em {ACM-SIAM} Symposium on Discrete Algorithms, {SODA} 2014},
  pages 82--101. {SIAM}, 2014.
\newblock \href {https://doi.org/10.1137/1.9781611973402.7}
  {\path{doi:10.1137/1.9781611973402.7}}.

\bibitem[Kla00]{Klazar}
Martin Klazar.
\newblock {The F{\"u}redi-Hajnal Conjecture Implies the Stanley-Wilf
  Conjecture}.
\newblock In {\em Formal Power Series and Algebraic Combinatorics: 12th
  International Conference, FPSAC’00}, pages 250--255. Springer, 2000.
\newblock \href {https://doi.org/10.1007/978-3-662-04166-6_22}
  {\path{doi:10.1007/978-3-662-04166-6_22}}.

\bibitem[Knu98]{Knuth3}
Donald~Ervin Knuth.
\newblock {\em The art of computer programming, , Volume III, 2nd Edition}.
\newblock Addison-Wesley, 1998.
\newblock URL: \url{https://www.worldcat.org/oclc/312994415}.

\bibitem[KS20]{smooth1}
L{\'{a}}szl{\'{o}} Kozma and Thatchaphol Saranurak.
\newblock Smooth heaps and a dual view of self-adjusting data structures.
\newblock {\em {SIAM} J. Comput.}, 49(5), 2020.
\newblock \href {https://doi.org/10.1137/18M1195188}
  {\path{doi:10.1137/18M1195188}}.

\bibitem[MT04]{MarcusTardos}
Adam Marcus and G{\'a}bor Tardos.
\newblock {Excluded permutation matrices and the Stanley--Wilf conjecture}.
\newblock {\em Journal of Combinatorial Theory, Series A}, 107(1):153--160,
  2004.
\newblock \href {https://doi.org/10.1016/J.JCTA.2004.04.002}
  {\path{doi:10.1016/J.JCTA.2004.04.002}}.

\bibitem[PR02]{PettieR02}
Seth Pettie and Vijaya Ramachandran.
\newblock An optimal minimum spanning tree algorithm.
\newblock {\em J. {ACM}}, 49(1):16--34, 2002.
\newblock \href {https://doi.org/10.1145/505241.505243}
  {\path{doi:10.1145/505241.505243}}.

\bibitem[ST85]{ST85}
Daniel~Dominic Sleator and Robert~Endre Tarjan.
\newblock Self-adjusting binary search trees.
\newblock {\em J. {ACM}}, 32(3):652--686, 1985.
\newblock \href {https://doi.org/10.1145/3828.3835}
  {\path{doi:10.1145/3828.3835}}.

\end{thebibliography}

\end{document}